\documentclass[11pt]{article}
\usepackage{amsfonts,amsmath,bm, xspace}
\usepackage{multirow,graphicx, algorithm,algorithmic,rotating}
\usepackage{url,cite}
\usepackage{fullpage}
\usepackage[small,bf]{caption}
\newtheorem{theorem}{Theorem}[section]
\newtheorem{lemma}[theorem]{Lemma}

\newtheorem{proposition}[theorem]{Proposition}
\newtheorem{definition}[theorem]{Definition}

\def \endprf{\hfill {\vrule height6pt width6pt depth0pt}\medskip}
\newenvironment{proof}{\noindent {\bf Proof} }{\endprf\par}

\numberwithin{equation}{section}

\newcommand{\R}{\mathbb{R}}

\newcommand{\sgn}{\textrm{sgn}}

\renewcommand{\P}{\operatorname{\mathbb{P}}}
\newcommand{\E}{\operatorname{\mathbb{E}}}

\newcommand{\cA}{\mathcal{A}}

\newcommand{\PT}{{\cal P}_T}
\newcommand{\PTc}{{\cal P}_{T^\perp}}

\newcommand{\vct}[1]{\bm{#1}}
\newcommand{\mtx}[1]{\bm{#1}}

\newcommand{\trace}{\operatorname{Tr}}

\newcommand{\minimize}{\mbox{minimize}}

\newcommand{\st}{\mbox{subject to}}

\numberwithin{equation}{section}

\newcommand{\eq}[1]{(\ref{eq:#1})}

\begin{document}

\title{Simple Bounds for Recovering Low-complexity Models}

\author{Emmanuel Cand\`{e}s\footnote{Mathematics and Statistics Departments, Stanford University, Stanford, CA 94305.  \texttt{candes@stanford.edu}} and Benjamin Recht\footnote{Computer Sciences Department, University of Wisconsin-Madison, Madison, WI, 53706.
\texttt{brecht@cs.wisc.edu}}}

\date{June 2011; revised February, 2012}

\maketitle

\vspace{-0.3in}

\begin{abstract}
This note presents a unified analysis of the recovery of simple objects from random linear measurements.  When the linear  functionals are Gaussian, we show that an   $s$-sparse vector in $\R^n$ can be efficiently recovered from $2s \log n$ measurements with high probability and a rank $r$, $n\times n$ matrix can be efficiently recovered from $r(6n-5r)$ measurements with high probability. For sparse vectors, this is within an additive factor  of the best known nonasymptotic bounds.  For low-rank matrices, this matches the best known bounds.  We present a parallel analysis for block-sparse vectors obtaining similarly tight bounds.   In the case of sparse and block-sparse signals, we additionally demonstrate that our bounds are only slightly weakened when the measurement map is a random sign matrix.  Our results are based on analyzing a  particular dual point which certifies optimality conditions of the respective convex programming problem. Our calculations rely only on standard large deviation inequalities and our analysis is self-contained.  
\end{abstract}

{\bf Keywords.} $\ell_1$-norm minimization \and nuclear-norm minimization \and block-sparsity \and duality \and random matrices.

\section{Introduction}

The past decade has witnessed a revolution in convex optimization algorithms for
recovering structured models from highly incomplete information.  Work
in compressed sensing has shown that when a vector is \emph{sparse},
then it can be reconstructed from a number of nonadaptive linear
measurements proportional to a logarithmic factor times the signal's
sparsity level~\cite{CRT06,DonohoCS}. Building on this work, many have
recently demonstrated that if an array of user data has low-rank, then
the matrix can be re-assembled from a sampling of information
proportional to the number of parameters required to specify a
low-rank factorization. See
\cite{Recht10,CandesRecht09,CandesPlanTight} for some early references
on this topic.

Sometimes, one would like to know precisely how many measurements are
needed to recover an $s$-sparse vector (a vector with at most $s$
nonzero entries) by $\ell_1$ minimization or a rank-$r$ matrix by
nuclear-norm minimization. This of course depends on the kind of
measurements one is allowed to take, and can be empirically determined
or approximated by means of numerical studies. At the theoretical
level, however, very precise answers---e.g., perfect knowledge of
numerical constants---for models of general interest may be very hard
to obtain.  For instance, in~\cite{CRT06}, the authors demonstrated
that about $20 s \log n$ randomly selected Fourier coefficients were
sufficient to recover an $s$-sparse signal, but determining the
minimum number that would suffice appears to be a very difficult
question. Likewise, obtaining precise theoretical knowledge about the
number of randomly selected entries required to recover a rank-$r$
matrix by convex programming seems delicate, to say the least. For
some special and idealized models, however, this is far easier and the
purpose of this note is to make this clear.

In this note, we demonstrate that many bounds concerning Gaussian measurements can be derived via elementary, direct methods using Lagrangian duality.  By a careful analysis of a particular Lagrange multiplier, we are able to prove that $2s \log n$ measurements are sufficient to recover an $s$-sparse vector in $\R^n$ and $r(6n-5r)$ measurements are sufficient to recover
a rank $r$, $n\times n$ matrix with high probability.  These almost match the best-known, non-asymptotic bounds for sparse vector reconstruction ($2s \log(n/s)+5/4s$ measurements~\cite{DonohoTanner09,CRPW10}), and match the best known bounds for low-rank matrix recovery in the nuclear norm (as reported in~\cite{CRPW10,Oymak10}).  

The work \cite{CRPW10}, cited above, presents a unified view of the convex programming approach to inverse problems and provides a relatively simple framework to derive exact, robust recovery bounds for a variety of simple models. As we already mentioned, the authors
also provide rather tight bounds on sparse vector and low-rank matrix
recovery in the \emph{Gaussian measurement ensemble} by using a deep
theorem in functional analysis due to Gordon, which concerns the
intersection of random subspaces with subsets of the
sphere~\cite{Gordon88b}.  Gordon's Theorem has also been used to provide sharp estimates of the \emph{phase transitions} for the $\ell_1$ and nuclear norm heuristics in~\cite{Stojnic09} and~\cite{Oymak10} respectively. Our work complements these results, demonstrating that the dual multiplier ansatz proposed in~\cite{Fuchs04} can also yield very tight bounds for many signal recovery problems.

To introduce our results, suppose we are given information about an
object $\vct{x}_0 \in \R^n$ of the form $\Phi \vct{x}_0
\in \R^m$ where $\Phi$ is an $m \times n$ matrix. When $\Phi$
has entries i.i.d.~sampled from a Gaussian distribution with mean $0$
and variance $1/m$, we call it a \emph{Gaussian measurement map}. We
want bounds on the number of rows $m$ of $\Phi$ to ensure that
$\vct{x}_0$ is the unique minimizer of the problem
\begin{equation}\label{eq:atomic-norm-primal}
	\begin{array}{ll} \minimize & \|\vct{x}\|_{\cA}\\
	\st & \Phi \vct{x} = \Phi \vct{x}_0. 
	\end{array}
\end{equation}
Here $\|\cdot\|_{\cA}$ is a norm with some suitable properties which
encourage solutions which conform to some notion of simplicity. Our
first result is the following 
\begin{theorem}
\label{teo:l1}
  Let $\vct{x}_0$ be an arbitrary $s$-sparse vector and
  $\|\cdot\|_{\cA}$ be the $\ell_1$ norm.  Let $\beta>1$.
  \begin{itemize}
  \item For Gaussian measurement maps $\Phi$ with $m \geq 2 \beta s
    \log n + s$, the recovery is exact with
    probability at least $1 - 2n^{-f(\beta,s)}$ where 
\[
f(\beta,s) = \Biggl[ \sqrt{\frac{\beta}{2s}+\beta-1
  }-\sqrt{\frac{\beta}{2s}}\Biggr]^2\,.
\]    
  \item Let  $\epsilon\in (0,1)$.  For binary measurement maps $\Phi$ with i.i.d.~entries taking on values $\pm m^{-1/2}$ with equal probability, there exist numerical constants $c_0$ and $c_1$ such that if $n\geq \exp(c_0/\epsilon^2)$ and $m \geq 2 \beta(1-\epsilon)^{-2} s\log n + s$, the recovery is exact with probability at least  $1-n^{1-\beta} - n^{-c_1 \beta \epsilon^2}$.
  \end{itemize}
\end{theorem}
The algebraic expression $f(\beta,s)$ is positive for all $\beta>1$
and $s>0$.  For all fixed $\beta>1$, $f(\beta,s)$ is an increasing
function of $s$ so that $\min_{s\geq 1} f(\beta,s) = f(\beta,1)$.
Moreover, observe that $\lim_{s \rightarrow \infty} f(\beta,s) =
\beta-1$. For binary measurement maps, our result states that for any $\delta>0$, $(2+\delta)s\log n$ entries suffice to recover an $s$-sparse signal when $n$ is sufficiently large. We also provide a very similar result for \emph{block-sparse} signals, stated in Section~\ref{sec:block-sparse}.

Our third result concerns the recovery of a low-rank matrix.
\begin{theorem}
\label{teo:nuclear}
Let $\mtx{X}_0$ be an arbitrary $n_1 \times n_2$ rank-$r$-matrix and
$\|\cdot\|_{\cA}$ be the matrix nuclear norm. For a Gaussian
measurement map $\Phi$ with $m \geq \beta r(3n_1 + 3n_2 - 5r)$ for
some $\beta>1$, the recovery is exact with probability at least
$1-2e^{(1-\beta)n/8}$, where $n = \max(n_1,n_2)$.
\end{theorem}
Our results are 1) nonasymptotic and 2) demonstrate sharp constants
for sparse signal and low-rank matrix recovery, perhaps the two most
important cases in the general model reconstruction
framework. Further, our bounds are proven using elementary concepts
from convex analysis and probability theory.  In fact, the most
elaborate result from probability that we employ concerns the largest
singular value of a Gaussian random matrix, and this is only needed to
analyze the rank minimization problem.

We show in Section~\ref{sec:generic-setup} that the same construction
and analysis can be applied to prove Theorems \ref{teo:l1} and
\ref{teo:nuclear}. The method, however, handles a variety of
complexity regularizers including the $\ell_1/\ell_2$-norm as well.
When specialized in Section~\ref{sec:bounds}, we demonstrate sharp
constants for exact model reconstruction in all three of these cases
($\ell_1$, $\ell_1/\ell_2$ and nuclear norms).  We conclude the paper
with a brief discussion of how to extend these results to other
measurement ensembles.  Indeed, with very minor modifications, we can
achieve almost the same constants for subgaussian measurement
ensembles in some settings such as sign matrices as reflected by the second part of Theorem
\ref{teo:l1}.

\section{Dual Multipliers and Decomposable
  Regularizers}\label{sec:generic-setup}


\begin{definition}
The dual norm is defined as
\begin{equation}\label{eq:dual-norm-def}
  \|\vct{x}\|_\cA^* = \sup \left\{\langle \vct{x}, \vct{a} \rangle ~ : ~ \|\vct{a}\|_{\cA}\leq 1  \right\}. 
\end{equation}
\end{definition}
A consequence of the definition is the well-known and useful dual-norm
inequality
\begin{equation}
  \label{eq:holder}
  |\langle \vct{x},\vct{y}\rangle| \leq \|\vct{x}\|_\cA \|\vct{y}\|_\cA^*.
\end{equation}
The supremum in \eq{dual-norm-def} is always achieved and thus the
dual norm inequality \eq{holder} is tight in the sense that for any
$\vct{x}$, there is a corresponding $\vct{y}$ that achieves equality.
Additionally, it is clear from the definition that the subdifferential
of $\|\cdot\|_\cA$ at $\vct{x}$ is $\{\vct{v} : \langle
\vct{v},\vct{x}\rangle = \|\vct{x}\|_\cA,\, \|\vct{v}\|_{\cA}^*\leq
1\}$.

\subsection{Decomposable Norms}

We will restrict our attention to norms whose subdifferential has very
special structure effectively penalizing ``complex'' solutions.  In a
similar spirit to~\cite{Negahban09}, the following definition
summarizes the necessary properties of a good complexity regularizer:
\begin{definition}
\label{def:decomposable}
  A norm $\|\cdot\|_{\cA}$ is~\emph{decomposable} at $\vct{x}_0$ if
  there is a subspace $T\subset \R^n$ and a vector $\vct{e}\in T$ such
  that the subdifferential at $\mtx{x}_0$ has the form
\begin{equation*}
  \partial\|\vct{x}_0\|_\cA = \left\{\vct{z}\in\R^n~:~\PT(\vct{z}) = \vct{e}~~\mbox{and}~~ \|\PTc(\vct{z})\|_{\cA}^*\leq 1\right\}
\end{equation*}
and for any $\vct{w}\in T^\perp$, we have
\begin{equation*}
  \|\vct{w}\|_{\cA} = \sup_{\stackrel{\vct{v}\in T^\perp}{\|\vct{v}\|_{\cA}^*\leq 1}} \langle \vct{v},\vct{w} \rangle\,.
\end{equation*} 
Above, $\PT$ (resp.~$\PTc$) is the orthogonal projection onto $T$
(resp.~orthogonal complement of $T$). 
\end{definition}

When a norm is decomposable at $\vct{x}_0$, the norm essentially
penalizes elements in $T^\perp$ independently from $\vct{x}_0$.  The
most common decomposable regularizer is the $\ell_1$ norm on $\R^n$.
In this case, if $\vct{x}_0$ is an $s$-sparse vector, then $T$ denotes
the set of coordinates where $\vct{x}_0$ is nonzero and $T^\perp$ the
complement of $T$ in $\{1,\ldots,n\}$. We denote by $\vct{x}_{0,T}$
the restriction of $\vct{x}_0$ to $T$ and by
$\operatorname{\sgn}(\vct{x}_{0,T})$ the vector with $\pm 1$ entries
depending upon the signs of those of $\vct{x}_{0,T}$.  The dual norm
to the $\ell_1$ norm is the $\ell_\infty$ norm.  The subdifferential
of the $\ell_1$ norm at $\vct{x}_0$ is given by
\begin{equation*}
  \partial\|\vct{x}_0\|_1 = \left\{\vct{z}\in\R^n~:~\PT(\vct{z}) 
= \operatorname{\sgn}(\vct{x}_{0,T})~\mbox{and}~ 
\|\PTc(\vct{z})\|_\infty \le 1 \right\}\,.
\end{equation*}
That is, $\vct{z}$ is equal to the sign of $\vct{x}_0$ on $T$ and has
entries with magnitudes bounded above by $1$ on the orthogonal
complement.  As we will discuss in Section~\ref{sec:bounds}, the
$\ell_1/\ell_2$ norm and the matrix nuclear norm are also
decomposable.  The following Lemma gives conditions under which
$\vct{x}_0$ is the unique minimizer of~\eq{atomic-norm-primal}.

\begin{lemma}\label{lemma:unicity} Suppose that $\Phi$ is injective
  on the subspace $T$ and that there exists a vector $\vct{y}$ in the
  image of $\Phi^*$ (the adjoint of $\Phi$) obeying
\begin{enumerate}
\item $\PT(\vct{y}) = \vct{e}$, where $\vct{e}$ is as in Definition
  \ref{def:decomposable}, 
\item $\|\PTc(\vct{y})\|_{\cA}^*< 1$. 
\end{enumerate}
Then $\mtx{x}_0$ is the unique minimizer of~\eq{atomic-norm-primal}.
\end{lemma}

\begin{proof}
  The proof is an adaptation from a standard argument.  Consider any
  perturbation $\vct{x}_0 + \vct{h}$ where $\Phi\vct{h}= \vct{0}$.
  Since the norm is decomposable, there exists a $\vct{v}\in T^\perp$
  such that $\|\vct{v}\|_\cA^*\leq 1$ and $\langle \vct{v},
  \PTc(\vct{h}) \rangle = \|\PTc(\vct{h})\|_\cA$.  Moreover, we have
  that $\vct{e}+\vct{v}$ is a subgradient of $\|\cdot\|_{\cA}$ at
  $\vct{x}_0$.  Hence,
\begin{align*}
\|\vct{x}_0+\vct{h}\|_{\cA} &\geq \|\vct{x}_0\|_{\cA} + \langle \vct{e}+\vct{v},\vct{h}\rangle\\
&= \|\vct{x}_0\|_{\cA} + \langle \vct{e}+\vct{v} - \vct{y},\vct{h}\rangle\\
&= \|\vct{x}_0\|_{\cA} + \langle\vct{v} - \PTc(\vct{y}),\PTc(\vct{h})\rangle\\
&\geq  \|\vct{x}_0\|_{\cA} + (1-\|\PTc(\vct{y})\|_{\cA}^*)\|\PTc(\vct{h})\|_{\cA}\,.
\end{align*}
Since $\|\PTc(\vct{y})\|_{\cA}^*$ is strictly less than one, this last inequality holds strictly unless $\PTc(\vct{h})=0$.  But if $\PTc(\vct{h})=0$, then $\PT(\vct{h})$ must also be zero because we have assumed that $\Phi$ is injective on $T$.  This means that $\vct{h}$ is zero proving that $\vct{x}_0$ is the unique minimizer of~\eq{atomic-norm-primal}.
\end{proof}

\subsection{Constructing a Dual Multiplier}\label{sec:dual-multiplier}

To construct a $\mtx{y}$ satisfying the conditions of Lemma~\ref{lemma:unicity}, we follow the program developed in~\cite{Fuchs04} and followed by many researchers in the compressed sensing literature.  Namely, we choose the least squares solution of $\PT(\Phi^* \vct{q})=\vct{e}$, and then prove that $\vct{y} := \Phi^* \vct{q}$ has dual norm strictly less than $1$ on $T^\perp$.  

Let $\Phi_T$ and $\Phi_{T^\perp}$ denote the restriction of
$\Phi$ to $T$ and $T^\perp$ respectively. Let $d_T$ denote the
dimension of the space $T$.  Observe that if $\Phi_T$ is injective,
then
\begin{align}\label{eq:q-def}
  \vct{q} &=\Phi_T (\Phi_T^* \Phi_T)^{-1} \vct{e},\\
\label{eq:y-def}	\PTc(\vct{y}) &= \Phi_{T^\perp}^* \vct{q}.
\end{align}
The key fact we use to derive our bounds in this note is that, when
$\Phi$ is a Gaussian map, $\vct{q}$ and $\Phi_{T^\perp}^*$ are
independent, no matter what $T$ is.  This follows from the isotropy of
the Gaussian ensemble. This property is also true in the sparse-signal
recovery setting whenever the columns of $\Phi$ are
independent. Another way to express the same idea is that given the
value of $\vct{q}$, one can infer the distribution of $\PTc(\vct{y})$
with no knowledge of the values of the matrix $\Phi_T$.

We assume in the remainder of this section that $\Phi$ is a
Gaussian map.  Conditioned on $\vct{q}$, $\PTc(\vct{y})$ is
distributed as
\begin{equation*}
  \iota_{T^\perp} \vct{g}, 
\end{equation*}
where $\iota_{T^\perp}$ is an isometry from $\R^{n-d_T}$ onto
$T^\perp$ and $\vct{g}\sim
\mathcal{N}(0,\frac{\|\vct{q}\|_2^2}{m}\mtx{I})$ (here and in the
sequel, $\|\cdot\|_2$ is the $\ell_2$ norm). Also, $\Phi_T$ is
injective as long as $m\geq d_T$ and to bound the probability that the
optimization problem~\eq{atomic-norm-primal} recovers $\vct{x}_0$, we
therefore only need to bound
\begin{equation}\label{eq:conditional-bound}
	\P[ \|\PTc(\vct{y})\|_{\cA}^* \geq 1] \leq 
	\P[ \|\PTc(\vct{y})\|_{\cA}^* \geq 1~|~\|\vct{q}\|_2 \leq \tau]  +  
	\P[ \|\vct{q}\|_2 \geq \tau] 
\end{equation}
for some value of $\tau$ greater than $0$.  The first term in the
upper bound will be analyzed on a case-by-case basis in
Section~\ref{sec:bounds}.  As we have remarked, once we have
conditioned on $\vct{q}$, this term just requires us to analyze the
large deviations of Gaussian random variables in the dual norm.  What
is more surprising is that the second term can be tightly upper
bounded in a generic fashion for the Gaussian ensemble, independent of
the regularizer under study.

To see this, observe that $\vct{q}$ has squared norm
\begin{equation*}
  \|\vct{q}\|_2^2 = \langle \vct{e},  (\Phi_T^* \Phi_T)^{-1} \vct{e}\rangle\,.
\end{equation*} 
By assumption, $ (\Phi_T^* \Phi_T)^{-1}$ is a $d_T\times d_T$ inverse
Wishart matrix with $m$ degrees of freedom and covariance $m^{-1}
\mtx{I}_{d_T}$.  Since the Gaussian distribution is isotropic, we have
that $\|\vct{q}\|_2^2$ is distributed as $\|\vct{e}\|_2^2 m B_{11}$,
where $B_{11}$ is the first entry in the first column of an inverse
Wishart matrix with $m$ degrees of freedom and covariance
$\mtx{I}_{d_T}$.

To estimate the large deviations of $\|\vct{q}\|_2$, it thus suffices
to understand the large deviations of $B_{11}$.  A classical result in
statistics states that $B_{11}$ is distributed as an inverse
chi-squared random variable with $m-d_{T}+1$ degrees of freedom~(see,
\cite[page 72]{MardiaBook} for example)\footnote{The reader not
  familiar with this result can verify with linear algebra that
  $1/B_{11}$ is equal to the squared distance between the first column
  of $\Phi_T$ and the linear space spanned by all the others.  This
  squared distance is a chi-squared random variable with $m-d_T+1$
  degrees of freedom.}. We can thus lean on tail bounds for the
chi-squared distribution to control the magnitude of $B_{11}$.  For
each $t > 0$,
\begin{equation}\label{eq:q-bound}
\begin{aligned}
  \P\left[\|\vct{q}\|_2\geq \sqrt{\frac{m}{m-d_T+1-t}} \|\vct{e}\|_2\right]&= \P[z \leq m-d_T+1-t] \\
  &\leq \exp\left(-\frac{t^2}{4(m-d_T+1)}\right)\,.
  \end{aligned}
\end{equation}
Here $z$ is a chi-squared random variable with $m-d_T+1$ degrees of
freedom, and the final inequality follows from the standard tail bound
for chi-square random variables (see, for example,~\cite{LaurentMassart}).


To summarize, we have proven the following
\begin{proposition}\label{thm:generic-bound}
  Let $\|\cdot\|_{\cA}$ be a decomposable regularizer at $\vct{x}_0$
  and let $t>0$.  Let $\vct{q}$ and $\vct{y}$ be defined as
  in~\eq{q-def} and~\eq{y-def}. Then $\vct{x}_0$ is the unique optimal
  solution of~\eq{atomic-norm-primal} with probability at least
  \begin{equation}~\label{eq:generic-bound} 1 - \P\left[
      \|\PTc(\vct{y})\|_{\cA}^* \geq 1~\bigg|~\|\vct{q}\|_2 \leq
      \sqrt{\tfrac{m}{m-d_T+1-t}} \|\vct{e}\|_2\right] -
    \exp\left(-\tfrac{t^2/4}{m-d_T+1}\right)\,.
\end{equation}
\end{proposition}

\section{Bounds}\label{sec:bounds}

Using Proposition~\ref{thm:generic-bound}, we can now derive
non-asymptotic bounds for exact recovery of sparse vectors, block-sparse vectors, and low-rank matrices in a unified fashion.

\subsection{Compressed Sensing in the Gaussian Ensemble}\label{sec:sparsity}

Let $\vct{x}_0$ be an $s$-sparse vector in $\R^n$.  In this case, $T$
denotes the set of coordinates where $\vct{x}_0$ is nonzero and
$T^\perp$ the complement of $T$ in $\{1,\ldots,n\}$. As previously
discussed, the dual norm to the $\ell_1$ norm is the $\ell_\infty$
norm and the subdifferential of the $\ell_1$ norm at $\vct{x}_0$ is
given by
\begin{equation*}
  \partial\|\vct{x}_0\|_1 = \left\{\vct{z}\in\R^n~:~\PT(\vct{z}) = \operatorname{\sgn}(\vct{x}_{0,T})~\mbox{and}~ \|\PTc(\vct{z})\|_\infty \le 1 \right\}\,.
\end{equation*}
Here, $\dim(T)=s$, the sparsity of $\vct{x}_0$, and
$\vct{e}=\sgn(\vct{x}_0)$ so that $\|\vct{e}\|_2=\sqrt{s}$.

For $m\geq s$, set $\vct{q}$ and $\vct{y}$ as in \eq{q-def} and
\eq{y-def}. To apply Proposition~\ref{thm:generic-bound}, we only need
to estimate the probability that $\|\PTc(\vct{y})\|_{\infty}$ exceeds
$1$ conditioned on the event that $\|\vct{q}\|_2$ is bounded.
Conditioned on $\vct{q}$, the components of $\PTc(\vct{y})$ in
$T^\perp$ are i.i.d.~$\mathcal{N}(0,\|\vct{q}\|^2_2/m)$. Hence, for
any $\tau>0$, the union bound gives
\begin{align}
\label{eq:inf-norm-union-bound}
  \P\left[ \|\PTc(\vct{y})\|_\infty \geq 1 ~|~ \|\vct{q}\|_2\leq \tau \right] & 
\le(n-s)  \P[|z| \ge \sqrt{m}/\tau] \nonumber\\
& \le n \exp\Bigl(-\frac{m}{2\tau^2}\Bigr),  
\end{align}
where $z \sim \mathcal{N}(0,1)$.  We have made use above of the
elementary inequality $\P(|z| \ge t) \le e^{-t^2/2}$ which holds for
all $t \ge 0$.  For $\beta>1$, select
\[
\tau = \sqrt{\frac{ms}{m-s+1-t}} \qquad \text{with} \qquad t = 2\beta
\log(n) \left(\sqrt{1+\frac{2s(\beta-1)}{\beta}}-1 \right)\,.
\]
Here, $t$ is chosen to make the two exponential terms in our probability equal to each other.
 We can put all of the parameters together and
plug~\eq{inf-norm-union-bound} into~\eq{generic-bound}.  For $m =
2\beta s \log n + s$, $\beta>1$, a bit of algebra gives the first part
of Theorem \ref{teo:l1}.

\subsection{Block-Sparsity in the Gaussian Ensemble}\label{sec:block-sparse}

In simultaneous sparse estimation, signals are block-sparse in the sense that $\R^n$ can be decomposed into a decomposition of subspaces
\begin{equation}\label{eq:Rn-decomp}
	\R^n = \bigoplus_{b=1}^M V_b
\end{equation}
with each $V_b$ having dimension $B$~\cite{Parvaresh08, Eldar09}.
We assume that signals of interest are only nonzero on a few of the $V_b$'s and search for a solution which minimizes the norm
\begin{equation*}
	\|\vct{x}\|_{\ell_1/\ell_2} = \sum_{b=1}^M \|\vct{x}_b\|_2,
\end{equation*}
where $\vct{x}_b$ denotes the projection of $\vct{x}$ onto $V_b$.

Suppose $\vct{x}_0$ is block-sparse with $k$ active blocks.  $T$ here
denotes the coordinates associated with the groups where $\vct{x}_0$
has nonzero energy.  $T^\perp$ is equal to all of the coordinates of
the groups where $\vct{x}_0=0$. The dual norm to the $\ell_1/\ell_2$
norm is the $\ell_\infty/\ell_2$ norm
\begin{equation*}
	\|\mtx{x}\|_{\ell_\infty/\ell_2} = \max_{1\leq b\leq M } \|\vct{x}_b\|_2\,.
\end{equation*}
The subdifferential of the $\ell_1/\ell_2$ norm at $\vct{x}_0$ is
given by
\begin{equation*}
  \partial\|\vct{x}_0\|_{\ell_1/\ell_2} = \left\{\vct{z}\in\R^n~:~\PT(\vct{z}) = \sum_{b\cap T \neq \emptyset} \frac{\vct{x}_{0,b}}{\|\vct{x}_{0,b}\|_2}~\mbox{and}~ \|\PTc(\vct{z})\|_{\ell_\infty/\ell_2} \le 1\right\}\,.
\end{equation*}
Much like in the $\ell_1$ case, $T$ denotes the span of the set of
active subspaces and $T^\perp$ is the set of inactive subspaces.  In
this formulation, $\dim(T)=kB$ and
\begin{equation*}
  \vct{e} = \sum_{b\cap T \neq \emptyset} \frac{\vct{x}_{0,b}}{\|\vct{x}_{0,b}\|_2}\,.
\end{equation*}
Note also that $\|\vct{e}\|_2= \sqrt{k}$.

With the parameters we have just defined, we can define $\vct{q}$ and
$\vct{y}$ by~\eq{q-def} and~\eq{y-def}.  If we again condition on
$\|\vct{q}\|_2$, the components of $\vct{y}$ on $T^\perp$ are
i.i.d.~$\mathcal{N}(0,\|\vct{q}\|_2^2/m)$. Using the union bound, we
have
\begin{equation}\label{eq:block-sparse-ub-chi}
  \P\left[ \|\PTc(\vct{y})\|_{\ell_\infty/\ell_2} \geq 1~|~\|\vct{q}\|_2\leq \tau \right] \leq  \sum_{b\in T^\perp} 	\P\left[\| \vct{y}_b\|_{2} \geq 1~|~\|\vct{q}\|_2\leq \tau \right]. 
\end{equation}
Conditioned on $\vct{q}$,
$\tfrac{m}{\|\vct{q}\|_2^2}\|\vct{y}_b\|_2^2$ is identically
distributed as a chi-squared random variable with $B$ degrees of
freedom.  Letting $u = \sqrt{\chi_B}$, the Borell inequality
\cite[Proposition 5.34]{VershyninRMT} gives
\[
\P(u \ge \E u + t) \leq e^{-t^2/2}.
\]
Since $\E u\le \sqrt{B}$, we have $\P(u \geq \sqrt{B} + t) \leq
e^{-t^2/2}$. Using this inequality, with
\[
\tau = \sqrt{\frac{m k}{m-kB+1-t}}, 
\]
we have that the probability of failure is upper bounded by 
\begin{equation}\label{eq:block-sparse-pf}
 M \exp\left(-\tfrac{1}{2} \left[\sqrt{\frac{m-kB+1 -t}{k}} -
    \sqrt{B}\right]^2\right) + 
\exp\left(-\frac{t^2/4}{m-kB+1}\right).
\end{equation}
Choosing $m \geq (1+\beta) k(\sqrt{B} +\sqrt{2\log M})^2 + k B $ and
setting $t = (\beta/2) k(\sqrt{B} +\sqrt{2\log M})^2$, we can then
upper bound (\ref{eq:block-sparse-pf}) by
\begin{multline*}
  M \exp\left(-\frac12 \left[\sqrt{1+\beta/2} (\sqrt{B} +\sqrt{2\log
        M})-
      \sqrt{B}\right]^2\right)  \\
  + \exp\left(-\frac{\beta^2}{16(1+\beta)} k (\sqrt{B} +\sqrt{2\log
      M})^2\right) \le M^{-\beta/4} +
  M^{-\beta^2/(8+8\beta)}\,. 
\end{multline*}
This proves the following
\begin{theorem}
\label{teo:blockL1}
Let $\vct{x}_0$ be a block-sparse signal with $M$ blocks of size $B$
and $k$ active blocks under the decomposition~\eq{Rn-decomp}. Let
$\|\cdot\|_{\cA}$ be the $\ell_1/\ell_2$ norm. For Gaussian
measurement maps $\Phi$ with
\[
m >  (1+\beta) k(\sqrt{B} +\sqrt{2\log M})^2 + k B
\]
the recovery is exact with probability at least $M^{-\beta/4} +
M^{-\beta^2/(8+8\beta)}$\,.
\end{theorem}
The bound on $m$ obtained by this theorem is identical to that
of~\cite{Rao12}, and is, to our knowledge, the tightest known
non-asymptotic bound for block-sparse signals.  For example, when
  the block size $B$ is much greater than $\log M$, the results
  asserts that roughly $2kB$ measurements are sufficient for the
  convex programming to be exact. Since there are $kB$ degrees of
  freedom, one can see that this is quite tight.

Note that the theorem gives a recovery result for sparse vectors by
setting $B=1$, $k=s$, and $M=n$.  In this case,
Theorem~\ref{teo:blockL1} gives a slightly looser bound and requires a
slightly more complicated argument as compared to
Theorem~\ref{teo:l1}.  However, Theorem~\ref{teo:blockL1} provides
bounds for more general types of signals, and we note that the same
analysis would handle other $\ell_1/\ell_p$ block regularization
schemes defined as $\|\vct{x}\|_{\ell_1/\ell_p} = \sum_{b = 1}^M
\|\vct{x}_b\|_p$ with $p \in [2,\infty]$. Indeed, the $\ell_1/\ell_p$
norm is decomposable and its dual is the $\ell_\infty/\ell_q$ norm
with $1/p + 1/q = 1$. The only adjustment would consist in bounding
$\|\vct{y}_{b}\|_q$; up to a scaling factor, this is a sum of
independent standard normals and our analysis goes through. We omit
the details.

\subsection{Low-Rank Matrix Recovery in the Gaussian
  Ensemble}\label{sec:lowrank}

To apply our results to recovering low-rank matrices, we need a little
bit more notation, but the argument is principally the same.  Let
$\mtx{X}_0$ be an $n_1 \times n_2$ matrix of rank $r$ with singular
value decomposition $\mtx{U}\mtx{\Sigma} \mtx{V}^*$.  Without loss of
generality, impose the conventions $n_1\leq n_2$, $\mtx{\Sigma}$ is
$r\times r$, $\mtx{U}$ is $n_1 \times r$, $\mtx{V}$ is $n_2 \times r$.

In the low-rank matrix reconstruction problem, the subspace $T$ is the
set of matrices of the form $\mtx{U}\mtx{Y}^* + \mtx{X}\mtx{V}^*$
where $\mtx{X}$ and $\mtx{Y}$ are arbitrary $n_1\times r$ and $n_2
\times r$ matrices.  The span of matrices of the form $\mtx{U}\mtx{Y}^*$ has dimension $n_1r$, 
 the span of $\mtx{X}\mtx{V}^*$ has dimension $n_2r$, and the intersection of these two spans has dimension $r^2$.  Hence, we have  $d_T=\operatorname{dim}(T) =
r(n_1+n_2-r)$.  $T^\perp$ is the subspace of matrices spanned by the
family $(\vct{x} \vct{y}^*)$, where $\vct{x}$ (respectively $\vct{y}$)
is any vector orthogonal to $\mtx{U}$ (respectively $\mtx{V}$). The
spectral norm denoted by $\| \cdot \|$ is dual to the nuclear norm.
The subdifferential of the nuclear norm at $\mtx{X}_0$ is given by
\begin{equation*}
	\partial \|\mtx{X}_0\|_* = \left\{\mtx{Z}~:~ \PT(\mtx{Z})=\mtx{U}\mtx{V}^*~\mbox{and}~\|\PTc(\mtx{Z})\|\leq 1\right\}\,.
\end{equation*}
Note that the Euclidean norm of $\mtx{U}\mtx{V}^*$ is equal to
$\sqrt{r}$.

For matrices, a Gaussian measurement map takes the form of a linear operator whose  $i$th component  is given by
\begin{equation*}
  [\Phi(\mtx{Z})]_i = \trace(\mtx{\Phi}_i^*\mtx{Z})\,.
\end{equation*}
Above, $\mtx{\Phi}_i$ is an $n_1\times n_2$ random matrix with i.i.d.,
zero-mean Gaussian entries with variance $1/m$.  This is equivalent to
defining $\Phi$ as an $m \times (n_1n_2)$ dimensional matrix acting on
$\operatorname{vec}(\mtx{Z})$, the vector composed of the columns of
$\mtx{Z}$ stacked on top of one another.  In this case, the dual
multiplier is a matrix taking the form
\begin{equation*}
	\mtx{Y} = \Phi^* \Phi_T (\Phi_T^* \Phi_T)^{-1}(\mtx{U}\mtx{V}^*)\,.
\end{equation*}
Here, $\Phi_T$ is the restriction of $\Phi$ to the subspace $T$.
Concretely, one could define a basis for $T$ and write out $\Phi_T$ as
an $m\times d_T$ dimensional matrix.  Note that none of the abstract
setup from Section~\ref{sec:dual-multiplier} changes for the matrix
recovery problem: $\mtx{Y}$ exists as soon as $m\geq
\operatorname{dim}(T)= r(n_1+n_2-r)$ and
$\PT(\mtx{Y})=\mtx{U}\mtx{V}^*$ as desired.  We need only guarantee
that $\|\PTc(\mtx{Y})\|< 1$.  We still have that
\begin{equation*}
	\PTc(\mtx{Y}) = \sum_{i=1}^m q_i \PTc(\mtx{A}_i) 
\end{equation*}
where $\vct{q} = \Phi_T (\Phi_T^* \Phi_T)^{-1}(\mtx{U}\mtx{V}^*)$ is
given by \eq{q-def} and, importantly, $\vct{q}$ and $\PTc(\mtx{A}_i)$
are independent for all $i$.

With such a definition, we can again straightforwardly
apply~\eq{generic-bound} once we obtain an estimate of
$\|\PTc(\mtx{Y})\|$ conditioned on $\vct{q}$. Observe that
$\PTc(\mtx{Y}) = {\cal P}_{U^\perp} \mtx{Y} {\cal P}_{V^\perp}$, where
${\cal P}_{U^\perp}$ (respectively ${\cal P}_{V^\perp}$ is a
projection matrix onto the orthogonal complement of $U$ (respectively
$V$).  It follows that $\PTc(\mtx{Y})$ is identically distributed to a
rotation of an $(n_1-r) \times (n_2-r)$ Gaussian random matrix whose
entries have mean zero and variance $\|q\|_2^2/m$.  Using the
Davidson-Szarek concentration inequality for the extreme singular
values for Gaussian random matrices~\cite{Davidson01}, we have
\begin{equation*}
	\P\left[ \|\PTc(\mtx{Y})\|>1~|~\|\vct{q}\|_2\leq \tau \right] \leq \exp\left(-\tfrac{1}{2} \left(\tfrac{\sqrt{m}}{\tau}-\sqrt{n_1-r}-\sqrt{n_2-r}  \right)^2\right)\,.
\end{equation*}

We are again in a position ready to prove~\eq{generic-bound}. To
guarantee matrix recovery, with $\tau = \sqrt{\frac{mr}{m - d_T + 1 -
    t}}$, we thus need
\begin{equation*}
\sqrt{\frac{m-d_T+1-t}{r}}-\sqrt{n_1-r}-\sqrt{n_2-r}\geq 0\,.
\end{equation*}
This occurs if
\begin{equation*}
	m \geq r(n_1+n_2-r) + \left(\sqrt{r(n_1-r)} + \sqrt{r(n_2-r)}\right)^2+t-1
\end{equation*}
But since $(a+b)^2\leq 2 (a^2+b^2)$, we can upper bound
\begin{equation*}
	\left(\sqrt{r(n_1-r)} + \sqrt{r(n_2-r)}\right)^2 \leq 2r(n_1+n_2-2r)\,.
\end{equation*}
Setting $t = (\sqrt{2r+1}-1) (\beta-1) (3n_1+3n_2 - 5r)$
in~\eq{generic-bound} then yields Theorem~\ref{teo:nuclear}.

\section{Discussion}


We note that with minor modifications, the results for sparse and
block-sparse signals can be extended to measurement matrices whose
entries are i.i.d.~subgaussian random variables.  In this case, we
can no longer use the theory of inverse Wishart matrices, but $\Phi_{T^\perp}$ and $\Phi_T$ are still independent, and we can bound the norm of $\mtx{q}$ using bounds on the smallest singular value of rectangular matrices.  For example, Theorem 39 in
\cite{VershyninRMT} asserts that there exist positive constants $\theta$
and $\gamma$ such that the smallest singular value obeys the deviation
inequality
\begin{equation}
\label{eq:vershyninRMT}
\P\left[\sigma_{\text{min}}(\Phi_T) \leq 1 - \theta \sqrt{d_T/m}
  -t \right] \leq e^{-\gamma m t^2}
\end{equation}
for $t > 0$.

We use this concentration inequality to prove the second part of
Theorem \ref{teo:l1}. Since $\|\Phi_T(\Phi_T^*
\Phi_T)^{-1}\| = \sigma_{\text{min}}^{-1}(\Phi_T)$, we have that 
\[
\|\vct{q}\|_2 \leq \frac{\sqrt{s}}{1-\theta\sqrt{s/m} - t} := \rho
\]
with probability at least $1-e^{-\gamma m t^2}$. This is the analog of
\eqref{eq:q-bound}. Now whenever $\|\vct{q}\|_2 \le \rho$, Hoeffding's
inequality \cite{1963} implies that \eqref{eq:inf-norm-union-bound}
still holds. Thus, we are in the same position as before, and obtain
\[
 \P\left[ \|\PTc(\vct{y})\|_\infty \geq 1 \right] \le 2(n-s) \exp\Bigl(-\frac{m}{2\rho^2}\Bigr) + \exp(-\gamma m t^2).
\]
Setting $t  =\epsilon/2$ proves the second part of Theorem \ref{teo:l1}.

For block-sparse signals, a similar argument would apply. The only
caveat is that we would need the following concentration bound which
follows from Lemma 5.2 in~\cite{Achlioptas03}: Let $\mtx{M}$ be an
$d_1\times d_2$ dimensional matrix with i.i.d. entries taking on
values $\pm 1$ with equal probability.  Let $\mtx{v}$ be a fixed
vector in $\R^{d_2}$. Then
\begin{equation*}
	 	\P\left[\| \mtx{M}\vct{v} \|_{2} \geq 1 \right]  \leq \exp\left(-\tfrac{\|v\|_2^{-2}-d_1}{4}\right)
\end{equation*}
provided $\|v\| \leq \sqrt{d_1}$.  Plugging this bound into~\eq{block-sparse-ub-chi} gives an analogous threshold for block-sparse signals in the Bernoulli model:

\begin{theorem}
\label{teo:block-sparse-bernoulli}
Let $\vct{x}_0$ be a block-sparse signal with $M$ blocks and $k$ active blocks  under the decomposition~\eq{Rn-decomp}. Let $\|\cdot\|_{\cA}$ be the $\ell_1/\ell_2$ norm.   Let $\beta>1$ and  $\epsilon\in (0,1)$.  For binary measurement maps $\Phi$ with i.i.d.~entries taking on values $\pm m^{-1/2}$ with equal probability, there exist numerical constants $c_0$ and $c_1$ such that if $M\geq \exp(c_0/\epsilon^2)$ and $m \geq  4k\beta(1-\epsilon)^{-2}\log M + 2kB$, the recovery is exact with probability at least  $1-M^{1-\beta} - M^{-c_1 \beta \epsilon^2}$.
\end{theorem}

For low-rank matrix recovery, the situation is more delicate.  With
general subgaussian measurement matrices, we no longer have
independence between the action on the subspaces $T$ and $T^\perp$
unless the singular vectors somehow align serendipitously with the
coordinate axes.  In this case, it unfortunately appears that we need
to resort to more complicated arguments and will likely be unable to
attain such small constants through the dual multiplier without a
conceptually new argument.

\begin{small}
\subsection*{Acknowledgements}
We thank anonymous referees for suggestions and corrections that have
improved the presentation. EC is partially supported by NSF via
grants CCF-0963835 and the 2006 Waterman Award; by AFOSR under grant
FA9550-09-1-0643; and by ONR under grant N00014-09-1-0258. BR is partially supported by ONR award N00014-11-1-0723 and NSF award CCF-1139953.  

\bibliographystyle{abbrv}
\bibliography{dual_multiplier}
\end{small}

\end{document}